\documentclass[11pt]{article}
\usepackage[utf8]{inputenc}
\usepackage[american]{babel}
\usepackage{amssymb}
\usepackage{amsmath}
\usepackage{thmtools}
\usepackage{listings}
\usepackage{graphicx}
\usepackage{color}
\usepackage{tikz}
\usepackage{subfig}
\usepackage{framed}
\usepackage{xspace}
\usepackage{todonotes}
\usepackage{enumerate}
\usepackage{xargs}
\usepackage{nicefrac}
\usepackage{csquotes}
\usepackage{mathtools}  
\usepackage[ruled,noline,linesnumbered]{algorithm2e}
\usepackage[backend=bibtex,doi=false,isbn=false,url=false,firstinits=true,maxcitenames=2,maxbibnames=99]{biblatex}

\AtEveryBibitem{
    \clearlist{location}
    \clearfield{pages}
    \clearname{editor}
}

\usetikzlibrary{arrows}
\usetikzlibrary{fit}
\usetikzlibrary{shadows}
\usetikzlibrary{patterns}
\usetikzlibrary{matrix}
\usetikzlibrary{shapes}
\usetikzlibrary{calc}
\usetikzlibrary{positioning}
\usetikzlibrary{backgrounds}

\newtheorem{theorem}{Theorem}[section]
\newtheorem{lemma}[theorem]{Lemma}
\newtheorem{corollary}[theorem]{Corollary}

\newenvironment{proof}{\noindent {\bf Proof.}\ }{\qed\par\vskip 4mm\par}

\newcommand{\qed}{\hfill $\square$}

\newcommandx*{\LDAUOmicron}[2][1=@pkling_false]{\mathcal{O}\ifthenelse{\equal{#1}{small}}{\bigl(#2\bigr)}{\left(#2\right)}}
\newcommandx*{\LDAUomicron}[2][1=@pkling_false]{\mathrm{o}\ifthenelse{\equal{#1}{small}}{\bigl(#2\bigr)}{\left(#2\right)}}
\newcommandx*{\LDAUOmega}[2][1=@pkling_false]{\Omega\ifthenelse{\equal{#1}{small}}{\bigl(#2\bigr)}{\left(#2\right)}}
\newcommandx*{\LDAUomega}[2][1=@pkling_false]{\omega\ifthenelse{\equal{#1}{small}}{\bigl(#2\bigr)}{\left(#2\right)}}
\newcommandx*{\LDAUTheta}[2][1=@pkling_false]{\Theta\ifthenelse{\equal{#1}{small}}{\bigl(#2\bigr)}{\left(#2\right)}}
\newcommandx*{\set}[2][2=@pkling_false]{\left\{#1\ifthenelse{\equal{#2}{@pkling_false}}{}{\;\middle|\;#2}\right\}}

\DeclareMathOperator{\diam}{diam}

\DeclareMathOperator{\PoA}{PoA}
\DeclareMathOperator{\PoS}{PoS}
\DeclareMathOperator*{\argmin}{arg\,min}

\makeatletter
\xspaceaddexceptions{\check@icr}
\makeatother

\title{Quality of Service in Network Creation Games\thanks{This work was partially supported by the German Research Foundation (DFG) within the Collaborative Research Centre ``On-The-Fly Computing'' (SFB 901) and by the EU within FET project MULTIPLEX under contract no.\ 317532.}%
~~\thanks{An extended abstract of this paper has been accepted for publication in the proceedings of the 10th International Conference on Web and Internet Economics (WINE), available at www.springerlink.com \cite{wine2014qosncg}.}}

\author{
    Andreas Cord-Landwehr \and
    Alexander Mäcker \and\newline
    Friedhelm Meyer auf der Heide
}

\date{Heinz Nixdorf Institute \& Department of Computer Science\linebreak University of Paderborn, Germany}

\begin{document}

\maketitle

\begin{abstract}
Network creation games model the creation and usage costs of networks formed by $n$ selfish nodes.
Each node $v$ can buy a set of edges, each for a fixed price $\alpha > 0$.
Its goal is to minimize its private costs, i.e., the sum (SUM-game, Fabrikant et al., PODC 2003) or maximum (MAX-game, Demaine et al., PODC 2007) of distances from $v$ to all other nodes plus the prices of the bought edges.
The above papers show the existence of Nash equilibria as well as upper and lower bounds for the prices of anarchy and stability.
In several subsequent papers, these bounds were improved for a wide range of prices $\alpha$.
In this paper, we extend these models by incorporating quality-of-service aspects:
Each edge cannot only be bought at a fixed quality (edge length one) for a fixed price $\alpha$.
Instead, we assume that quality levels (i.e., edge lengths) are varying in a fixed interval $[\check\beta,\hat\beta]$, $0 < \check\beta \leq \hat\beta$.
A node now cannot only choose which edge to buy, but can also choose its quality $x$, for the price $p(x)$, for a given price function $p$.
For both games and all price functions, we show that Nash equilibria exist and that the price of stability is either constant or depends only on the interval size of available edge lengths.
Our main results are bounds for the price of anarchy.
In case of the SUM-game, we show that they are tight if price functions decrease sufficiently fast.
\end{abstract}

\section{Introduction}
Network creation games (NCG) aim to model the evolution and outcome of networks created by selfish nodes.
In these games, nodes can decide individually which edges they want to buy in order to minimize their private costs, i.e., the costs of the bought edges plus costs for communicating with other nodes.
Each node $v$ can buy a set of edges, each for a price $\alpha >0$.
Its goal is to minimize its private costs, i.e.,  the sum (SUM-game) or maximum (MAX-game) of the distances from $v$ to all other nodes  in the network plus the costs of the bought edges.
Since all decisions are taken individually and only with respect to optimize their private costs, analyzing the resulting network by comparing it to an overall good structure constitutes the central aspect in the study of NCGs.
This task was formalized as analyzing the \emph{price of anarchy} and was first discussed by \textcite{fabrikant2003} for the SUM-game and by \textcite{demaine2007} for the MAX-game.
These papers inspired a series of subsequent works.

In this paper, we incorporate a kind of quality-of-service into the classical network creation games model:
An edge can be bought for different prices, with different latencies.
This is a well-established method, e.g., for internet service providers who offer different bandwidths of connection for different prices.
Considering the recent Netflix deal \cite{netflixdeal}, where the video streaming provider Netflix is paying an internet service provider to provide good connections to its customers, similar effects start to appear on the level of internet service providers.
We formalize this game on individual connection qualities and prove the existence of equilibria and present bounds for the prices of stability and anarchy.

\subsection{Model \& Notations}
An instance of our NCG is given by a set $V$ of $n$ nodes and a price function $p:[\check\beta,\hat\beta]\rightarrow\mathbb{R}^+$ on an interval of possible edge weights $[\check\beta,\hat\beta] \subseteq \mathbb{R}^+$.
A price function is assumed to be monotonically decreasing and the interval to fulfill $0 < \check\beta \leq \hat\beta$.
Each $v\in V$ aims to minimize its private costs by selfishly selecting a \emph{strategy} $s_v\subset V \times [\check\beta,\hat\beta]$ such that each $(u,x) \in s_v$ represents an undirected weighted edge $(\{v,u\}, x)$ from $v$ to $u$ of weight $x$, which is created by $v$ and has price $p(x)$.
For a strategy profile $S = (s_1,\ldots,s_n)$, the resulting weighted graph $G[S]$ consists of vertices $V$ and the weighted edges $\bigcup_{v\in V}\{(\{v,u\}, x)|(u,x)\in s_v\}$.

The costs of a node in the SUM-game are given by $c_v(S)=\sum_{(u,x)\in s_v} p(x) + \sum_{u\in V} d_{G[S]}(v,u)$.
Here, $d_{G[S]}(v,u)$ denotes the shortest weighted path distance from $v$ to $u$ in the weighted graph $G[S]$.
For the MAX-game, the private cost function is $c_v(S)=\sum_{(u,x)\in s_v} p(x) + \max_{u\in V} d_{G[S]}(v,u)$.

The social costs in both games are $c(S) = \sum_{v\in V}c_v(S)$.
We refer to the edge cost term of the cost function as \emph{edge costs} and to the distance term as \emph{distance costs}.
A strategy profile $S=(s_1,\ldots,s_n)$ is called a \emph{Nash equilibrium} (NE) if for every node $i$ and every strategy $s_i'$ it holds:
Let $s'_i\not=s_i$ be a strategy change, then $S'\coloneqq(s_1,\ldots,s_{i-1},s'_i,s_{i+1},\ldots,s_n)$ does not have lower costs for $i$, i.e, $c_i(S)\leq c_i(S')$.
Depending on the game, we call such an equilibrium a SUM-NE or a MAX-NE.
If a strategy profile is not a NE, then there exists at least one node that can perform an \emph{improving response} (IR), i.e., it can decrease its costs by changing its strategy.
An improving response is called a \emph{best response} (BR) if this strategy change is optimal regarding the maximum private costs decrease.

A main objective of the research on NCGs is the analysis of the \emph{price of anarchy}.
This measure gives the quality of Nash equilibria by comparing their social costs to the smallest social cost possible for the given price function.
The price of anarchy, introduced in \cite{koustoupias1999}, is defined as the ratio of a largest social cost  of any Nash equilibrium and the optimal social cost.
The minimal loss by selfish behavior is given by the \emph{price of stability}, see for example \cite{anshelevich2003,anshelevich2004}, and it is defined as the ratio of the smallest social cost of any Nash equilibrium and the optimal social cost.

\subsection{Related Work}
In the original SUM- and MAX-games by \textcite{fabrikant2003} and \textcite{demaine2007}, nodes can only buy edges of fixed length one for a fixed price $\alpha>0$.
In \cite{fabrikant2003}, the authors introduced the SUM-game and proved (among other things) an upper bound of $\LDAUOmicron{\sqrt{\alpha}}$ on the price of anarchy (PoA) in the case of $\alpha<n^2$, and a constant PoA otherwise.
Later, \textcite{albers2006} proved a constant PoA for $\alpha=\LDAUOmicron{\sqrt{n}}$ and the first sublinear worst case bound of $\LDAUOmicron{n^{1/3}}$ for general $\alpha$.
\textcite{demaine2007} were the first to prove an $\LDAUOmicron{n^{\varepsilon}}$ bound for $\alpha$ in the range of $\LDAUOmega{n}$ and $\LDAUomicron{n\lg n}$.
Recently, by \textcite{mihalak2010} and improved by \cite{mihalak2013treeEquilibria}, it was shown that for $\alpha \geq 65 n$ all equilibria are trees (and thus the PoA is constant).
For non-integral constant values of $\alpha > 2$, \textcite{hamilton2013anarchyIsFree} showed that the PoA tends to $1$ as $n\rightarrow\infty$.

For the MAX-game, \textcite{demaine2007} showed that the PoA is at most $2$ for $\alpha\geq n$, for $\alpha$ in range $2\sqrt{\lg n}\leq \alpha\leq n$ it is $\LDAUOmicron{\min\{4^{\sqrt{\lg n}},(n/\alpha)^{1/3}\}}$, and $\LDAUOmicron{n^{2/\alpha}}$ for $\alpha < 2\sqrt{\lg n}$.
For $\alpha>129$, \textcite{mihalak2010} showed, like in the SUM version, that all equilibria are trees and the PoA is constant.

In \textcite{alon2010}, a simpler model, the \emph{basic network creation game} (BNCG), was introduced.
Here, the operation of a node consists of swapping some of its incident edges, i.e., redirecting them to other nodes.
There are no costs associated with such operations.
Restricting the initial network to trees, the only equilibrium in the SUM-game is a star graph.
Without restrictions, all (swap) equilibria are proven to have a diameter of $2^{\LDAUOmicron{\sqrt{\log n}}}$, which is also the PoA.
For the MAX-version, the authors provide an equilibrium network with diameter $\LDAUTheta{\sqrt{n}}$.
In \textcite{lenzner2012greedy}, it is shown for the SUM-game that this model is fundamentally different to the original network creation game in the following sense:
There are equilibria for the BNCG that are not equilibria for NCG for any $\alpha$, and vice versa.

An interesting extension of the SUM-game was introduced and investigated by \textcite{albers2006}.
They  weight each pair $(u,v)$ of nodes, indicating the importance of the connection to $v$ for $u$. The special case of 0-1-weights, defining a friendship graph between the nodes, was examined by \textcite{halevi2007} and \textcite{cola2012ibncg}.

\subsection{Our Results}
We show that equilibria exist for both games, for every monotonically decreasing price function $p:[\check\beta,\hat\beta]\rightarrow\mathbb{R}^+$.
For the SUM-game with $n$ nodes, the price of stability is at most $\mathcal{O}(1 + \hat\beta/\check\beta)$.
For the MAX-game, it is always constant.

For the price of anarchy for the SUM-game, the upper bound is $\LDAUOmicron{\min\{n,(p(x^*) + x^*)/\check\beta\}}$, with $x^*\in[\check\beta,\hat\beta]$ being the edge quality that minimizes $p(x) + x$.
This value can be understood to be the edge weight with optimal price-weight trade-off for an edge, if used for exactly one shortest path.
For example, for $p:[\check\beta,\hat\beta]\rightarrow\mathbb{R}^+$ with $p(x)=\alpha/x$ and $\alpha > 0$, the price of anarchy is $\LDAUOmicron{\sqrt{\alpha}/\check\beta}$.

If  $x^*=\hat\beta$, $p(\hat\beta)\leq\check\beta$, and $p(\check\beta)\leq\hat\beta$ hold, then this upper bound is tight up to constant factors.
Examples for such price functions are linear functions $p:[1,\alpha-2\varepsilon]\rightarrow\mathbb{R}^+$ with $p(x)=\alpha - (1+\varepsilon)x$, for $\alpha > 0$ and  $0 < \varepsilon < \nicefrac{1}{2}$.
For these functions the price of anarchy is $\LDAUTheta{\alpha-\varepsilon}$.

For the MAX-game we provide a price of anarchy upper bound of $\LDAUOmicron{1 + \sqrt[3]{n}}$.

\section{The Sum-Game}
In this section, we consider the quality of equilibria in the SUM-game.

\begin{lemma}
\label{lemma:sumSocialCostBound}
Let $S$ be a strategy profile such that $G[S]$ is connected and no edge can be removed without increasing the social cost.
Denote by $\check x$ the minimal weight of any edge in $G[S]$ and by $m$ the number of all edges.
Then, for $x^*\in[\check\beta,\hat\beta]$ being the value minimizing $p(x) + x$, it holds
$c(S) \geq 2\check x n(n-1) + m(p(x^*) + x^* - 4\check x)$.
\end{lemma}
\begin{proof}
Let $E_x  \coloneqq  \bigcup_{v \in V} \{\{v,u\} | (u,x) \in s_v \}$ denote the edges in $G[S]$ of weight $x$.
Using that not directly connected nodes have (weighted) distance at least $2\check x$:
\begin{align*}
c(S) & \geq \sum_{x \in[\check\beta,\hat\beta]:|E_x|\not=0} (p(x) + 2 x)|E_x| + 2\check x \left(n(n-1)-2 \sum_{x \in[\check\beta,\hat\beta]:|E_x|\not=0}|E_x|\right) \\
    & \geq 2\check x n(n-1) + \sum_{x \in[\check\beta,\hat\beta]:|E_x|\not=0} |E_x|(p(x) + x - 4\check x)
     \geq 2\check x n(n-1) + m(p(x^*) + x^* - 4\check x)
\end{align*}
\end{proof}

\begin{lemma}
\label{lemma:sumOptimalSolutions}
Let $p:[\check\beta,\hat\beta]\rightarrow\mathbb{R}^+$ be a price function.
Define $\chi^*\in [\check\beta,\hat\beta]$ to be the value minimizing $p(x) + 2x$ and $\bar\chi \in [\check\beta,\hat\beta]$ the value minimizing $p(x) + 2(n-1)x$.
Then, the optimal social cost is given by a star with all edges having weight $\bar\chi$ or by a complete graph with all edges having weight $\chi^*$.
\end{lemma}
\begin{proof}
First, we argue that it is enough to only consider graphs with hop diameter of $\leq 2$ by converting any graph $G$ to a graph $G'$ with not higher social cost but hop diameter one or two.
For this, let $x_1\leq \ldots\leq x_m$ be the weights of all edges of $G$ (same weights are listed multiple times), ordered by increasing weight.
We create $G'$ with one center node $c$ and $n-1$ satellite nodes $v_1,\ldots,v_{n-1}$ and an initial edge set $E_c$ by connecting each $v_i$ to $c$ by an edge of weight $x_i$.

In the following, we create a one-to-one mapping from node pairs in $G$ to node pairs in $G'$ such that no distance is increased.
First, for every directly connected node pair $(u,v)$ such that its edge is associated to an edge in $E_c$, we map $u$ and $v$ to the end points of this edge.
Secondly, for every node pair $(u,v)$ that is not directly connected, but where both the first and last edge of a shortest path in $G$ are associated to edges in $E_c$, we map the nodes to the satellite nodes adjacent to the corresponding edges in $G'$.
Thirdly, for every node pair $(u,v)$ that is not directly connected in $G$, but where the first and not the last edge of a shortest path in $G$ is associated to an edge in $E_c$, we map to the satellite node adjacent to the corresponding edge and to an arbitrary other node to that not more than $n-2$ nodes are mapped (including the mappings done by edge associations during the current step).
Fourthly, for all remaining not directly connected node pairs in $G$, we map them to arbitrary pairs of nodes to which no mapping is performed yet.
Finally, for all directly connected edges in $G$, which are not yet mapped, map them to an arbitrary pair of nodes.
Here, if the distance of the mapped pair in $G'$ is at most the distance as in $G$ then do nothing.
Otherwise, create an edge between them of same weight as in $G$.

By construction, this mapping is bijective as all node pairs are mapped and to every node pair in $G'$ only one pair is mapped.
Further, no distance of a mapped pair of nodes is bigger than the corresponding distance in $G$.
This is obvious for the first and last step.
For the other steps, this holds by the fact that $x_1,\ldots,x_{n-1}$ are the minimal weights of all edges in $G$.
Finally, the edge costs in $G'$ are at most the edge costs in $G$ and hence $c(G')\leq c(G)$.

We claim that the so constructed graph either is a star or a complete graph.
By construction, a shortest distance path between any two not directly connected nodes $u,v\in V\setminus\{c\}$ must contain $c$.
Hence, any edge connecting two satellite nodes $u,v\in V\setminus\{c\}$ is used exclusively for the shortest paths $u$ to $v$ and $v$ to $u$ and thus has weight $\bar\chi$.
Using the social cost optimality, all satellite nodes must have the same degree and be connected by edges with weight of $x$ to $c$.
For $m$ being the total number of edges that connect any two satellites, the social cost is $m(p(\bar\chi) + 2\bar\chi - 4x) + (n-1)(p(x) + 2(n-1)x)$.
We see that for any fixed $x$ this term is minimized either with $m=0$ (lower bound for $m$) or $m=(n-1)n/2 - (n-1)$ (upper bound for $m$).
Hence, the optimal solution is either a star or a complete graph.
For a complete graph, all weights are $\chi^*$ with  $\chi^*$ minimizing the social cost given by $n(n-1)(x + p(x)/2)$.
Otherwise, for a star the edge weight $\bar\chi$ minimizes the social cost given by
$$
    2(n-1)x + (n-1)(n-2)2x + (n-1)p(x) = (n-1)(2(n-1)x + p(x))
$$
\end{proof}

Given a price function $p$, the \emph{price of stability} is the minimum loss of efficiency by the selfish acting of the nodes, i.e., the best case social cost ratio of an equilibrium network and an optimal solution.

\begin{theorem}
\label{thm:sumEqExistence}
For every price function $p:[\check\beta,\hat\beta]\rightarrow\mathbb{R}^+$ a SUM-NE exists.
Let $x^*\in [\check\beta,\hat\beta]$ be the value minimizing $p(x) + x$ and $\bar x\in [\check\beta,\hat\beta]$ the value minimizing $p(x) + (n-1)x$, then either a star with all edges having weight $\bar x$ or a complete graph with all edges having weight $x^*$ forms a SUM-NE graph.
\end{theorem}
\begin{proof}
We distinguish three cases, first depending on whether $p(x^*) > \bar x$ or not and then for $p(x^*) \leq \bar x$, whether $p(\bar x) < x^*$ or not.

For $\bar x < p(x^*)$ we consider a star graph consisting of a center node $c$ and $n-1$ satellite nodes $v_1,\ldots,v_{n-1}$.
Every satellite node $v_i$ owns one edge towards $c$ of weight $\bar x$.
Considering the center node, $c$ will not create any edge since it is directly connected to all other nodes and $\bar x\leq x^*$ (follows by definition), thus every edge weight is less than the optimal weight for connecting $c$ to exactly one node.
But also no satellite node $v_i$ will perform an operation, since on the one hand the weight of $v_i$'s only edge is optimal for being the only connection to $n-1$ nodes.
On the other hand, for the optimal cost $p(x^*)$ to improve the distance to exactly one other satellite node, the gain is $2\bar x - x^* - p(x^*) \leq 0$.
Hence the star forms a SUM-NE.

For $p(x^*) \leq \bar x$, we check when a clique with all edges having weight $x^*$ and arbitrarily assigned edge ownerships forms a SUM-NE graph.
First, we see that decreasing the weight of any edge $\{u,v\}$ exclusively decreases the distances between $u$ and $v$.
Since the edge weight $x^*$ is the optimal weight for an edge used for exactly one shortest path, no node will change any edge weight as well as create a new edge.
Hence, we only have to consider the IR of removing some edges and creating one new edge of weight $x \leq x^*$.
Since all strategy changes are unilateral, after this change the hop-diameter is two.
This is, other nodes are either directly connected or at distance $x + x^*$.
We consider the best response for a node $v$ consisting of (optimally) removing $n-1$ edges of weight $x^*$ and creating one new edge of weight $\bar x$ to an arbitrary node.
This changes the costs of $v$ by
$$
-(n-1)p(x^*) + p(\bar x) + (n-2)\bar x - (x^* - \bar x)
    = p(\bar x) - x^* + (n-1)(\bar x - p(x^*))
    \geq p(\bar x) - x^*
$$
In particular, for $p(\bar x)\geq x^*$ the costs increase.

It remains to consider the last case with $p(x^*)\leq \bar x$, $p(\bar x)< x^*$ and $p(\bar x) - x^* + (n-1)(\bar x - p(x^*)) < 0$.
Again, like in the first case, we claim that a star forms a SUM-NE graph.
By choosing all edge weights to be $\bar x$, we only have to show that creating a new edge gives no gain.
The optimal cost change by creating a new edge is $p(x^*) - 2\bar x + x^*$.
This is an IR if and only if $x^* < 2\bar x - p(x^*)$.
But combining both constrains gives:
\begin{align*}
0 & > p(\bar x) - x^* + (n-1)(\bar x - p(x^*))
    \geq p(\bar x) - (2\bar x - p(x^*)) + (n-1)(\bar x - p(x^*))\\
    & \geq p(\bar x) - 2\bar x + p(x^*) + (n-1)\bar x - (n-1)p(x^*)
    \geq (n-3)\bar x  - (n-3) p(x^*)\\
    & \geq (n-3)(\bar x - \bar x)
    \geq 0
\end{align*}
which is a contradiction and hence the star forms a SUM-NE.
\end{proof}

\begin{corollary}
\label{cor:sumPoS}
Let $x^*\in [\check\beta,\hat\beta]$ be the value minimizing $p(x) + x$ and $\bar x\in [\check\beta,\hat\beta]$ the value minimizing $p(x) + (n-1)x$.
Then, the price of stability in the SUM-game is constant if $p(x^*) > \bar x$ or $x^* > p(\bar x)$ and otherwise $\LDAUOmicron{1 + x^*/\bar x}$.
\end{corollary}
\begin{proof}
We consider the equilibrium graphs as given in Theorem~\ref{thm:sumEqExistence} and the optimal solution graphs as stated in Lemma~\ref{lemma:sumOptimalSolutions}.
First, consider the case that our equilibrium graph forms a star with all edges having weight $\bar x$.
If the optimal solution is also a star with all edges having weight of $\bar\chi$, we get $\PoS \leq 2$.
%
Otherwise, if the optimal solution forms a clique, we get:
\begin{align*}
    \PoS & \leq \frac{(n-1)(2(n-1)\bar x + p(\bar x)))}{n(n-1)(\chi^* + p(\chi^*)/2)}
            \leq \frac{2(n-1)\bar x + p(\bar x)}{n(x^* + p(x^*))/2}\\
        & = \frac{n-1}{n}\frac{4\bar x}{x^* + p(x^*)} + \frac{2p(\bar x)}{n(x^* + p(x^*))} \leq 6
\end{align*}
Here, the first term is at most $4$ with $x^* \geq \bar x$ and the second term at most $2$, since the star forms a SUM-NE and hence $n(x^* + p(x^*)) > p(\bar x)$ must hold.
(Otherwise any node owning an edge can perform an IR.)

Next, consider the case when the SUM-NE graph forms a clique.
If also the optimal solution is a clique, we get $\PoS \leq 2$ (analog to the first case.)
Otherwise, if the optimal solution forms a star, we can estimate the PoS as:
\begin{align*}
    \PoS & \leq \frac{n(n-1)(x^* + p(x^*)/2)}{(n-1)(2(n-1)\bar\chi + p(\bar\chi))}
            \leq \frac{n(x^* + p(x^*)/2)}{(n-1)\bar x + p(\bar x)}\\
        & \leq \frac{n p(x^*)}{(n-1)\bar x + p(\bar x)} + \frac{nx^*}{(n-1)\bar x + p(\bar x)}
            \leq 2 + 2x^*/\bar x
\end{align*}
Here, the first term comes from $p(x^*)\leq \bar x$.
\end{proof}

Similar to \textcite{albers2006}, we start our analysis for the price of anarchy by bounding the social cost of a SUM-NE graph essentially by the diameter of the graph.
Using arguments about the maximum weights and prices in SUM-NE graphs, we can further bound this diameter and get a PoA upper bound only depending on the price function and its domain.

\begin{lemma}
\label{le:neCost}
Let $p:[\check\beta,\hat\beta]\rightarrow\mathbb{R}^+$ be a price function and $S$ a SUM-NE strategy profile.
Then
$$
    c(S) \leq n\delta_{G}(v) + x^*(n-1)^2 + 2 (p(x^*)+x^*)n (n-1)
$$
with $G \coloneqq G[S],\delta_{G}(v) \coloneqq \sum_{u\in V} d_{G}(v,u)$ being the distance costs of an arbitrary node in the equilibrium graph, and $x^*\in[\check\beta,\hat\beta]$ chosen such that it minimizes $p(x) + x$.
\end{lemma}
\begin{proof}
First, we show that in an equilibrium graph all edges have a price of at most $n(p(x^*)+x^*)$.
For this, assume there is an edge of price $p(x) > (p(x^*)+x^*)n$.
Then, exchanging this edge with one having weight $x^*$ would decrease the edge costs by $p(x)-p(x^*) > nx^*+(n-1)p(x^*)$ while increasing the distance costs only by at most $(x^*-x) (n-1)$.
Since $(x^*-x)(n-1) < nx^* + (n-1)p(x^*)$, this is an improving response and hence contradicts $S$ forming a SUM-NE.

Next, take an arbitrary node $v$ and consider a shortest path tree $T$ of $v$ in $G$.
For every $u\in V$ define $m_u  \coloneqq  |\{\{u,w\} | (w,x) \in s_u \wedge \{u,w\} \in T\}|$ to be the number of tree edges maintained by node $u$.
Then, for the cost of any node $u \neq v$ it holds
$$
    c_u(S) \leq (p(x^*)+x^*)n(m_u + 1) + \delta_{G}(v) + x^*(n-1)
$$
To see this, we assume the contrary and define an improving response for $u$ in which $u$ removes all own edges, except those belonging to $T$, and additionally creates one edge of weight $x^*$ to $v$.
This new strategy incurs edge costs of at most $(p(x^*)+x^*)n m_u + p(x^*)$ for $u$.
Since by this strategy change $v$'s distance costs are not changed, $u$'s distance costs are at most $\delta_{G}(v) + (n-1)x^*$.
Hence, by having $S$ forming a SUM-NE, the original private costs of $u$ cannot be higher than claimed.

Using this bound for every node $u \ne v$ and the fact that $v$ only owns edges belonging to $T$ (removing a non-tree edge would reduce $v$'s cost), we have:
\begin{align*}
c(S) &\leq \delta_{G}(v) + (p(x^*) + x^*)n m_v\\
    &\phantom{=}    + \sum_{u \neq v}((p(x^*)+x^*)n(m_u + 1) + \delta_{G}(v) + x^*(n-1))\\
    &= n \delta_{G}(v) + x^*(n-1)^2 + (p(x^*)+x^*)n m_v +\sum_{u \neq v}{(p(x^*)+x^*)n(m_u + 1)} \\
    & = n\delta_{G}(v) + x^*(n-1)^2 + 2 (p(x^*)+x^*)n (n-1)
\end{align*}
For the last equality we use that a tree with $n$ nodes has $n-1$ edges.
\end{proof}

\begin{lemma}
\label{lemma:SumNEboundByDiameter}
Let $p:[\check\beta,\hat\beta]\rightarrow\mathbb{R}^+$ be a price function and $S$ a strategy profile forming a SUM-NE graph $G$.
Then, the diameter of $G$ is at most $\LDAUOmicron{p(x^*)+x^*}$, with $x^*\in[\check\beta,\hat\beta]$ minimizing $p(x) + x$.
\end{lemma}
\begin{proof}
First, we show that no edge can have weight bigger than $p(x^*) + x^*$.
Assuming there is $(v,x)\in s_u$ such that the corresponding edge $\{u,v\}$ has weight $x > p(x^*)+x^*$, we consider the exchange of this edge with an edge of weight $x^*$.
Then, the new strategy $s'_u  \coloneqq  (s_u \setminus \{(v,x)\}) \cup \{(v,x^*)\}$ decreases the distance costs by at least $x - x^* > p(x^*)$ while increasing the edge costs by only $p(x^*)-p(x)$.
Since this contradicts $G$ being a SUM-NE graph, we conclude the claim.

Next, we consider the length of a longest shortest path in $G$.
Denote the end points of such a path by $u$ and $v$.
If this path consists of only one edge, this edge would have weight of at most $p(x^*)+x^*$, and the claim holds.
If the path consists of at least two edges, we define a parameter $k\in\mathbb{R}^+$ such that $2k = d_{G}(u,v)$ and consider the strategy change $s'_u  \coloneqq  s_u \cup \{(v,x)\}$ of node $u$ by creating the edge $\{u,v\}$ with $x$ being an arbitrary but fixed weight $x \in [\check\beta,\hat\beta]$.
This change decreases the distances from $u$ to all nodes on the path that have distance of at least $k+x$ to $u$.
Let $v \eqqcolon v_1, v_2, \ldots, v_{Z}$ denote these nodes, ordered by increasing distance to $v$.
Since each edge has weight of at most $\min\{p(x^*)+x^*, \hat\beta\}$, we get $Z \geq \left\lceil \frac{k-x}{\min\{p(x^*)+x^*, \hat\beta\}} \right\rceil$.
With the strategy change $s'_u$, each distance from $u$ to $v_i$ decreases from $2k - d_{G}(v,v_i)$ to be at most $x + d_{G}(v,v_i)$, resulting in a distance cost decrease of at least:
\begin{align*}
&\sum_{i=1}^Z(2k-d_{G}(v,v_i)) - \sum_{i=1}^Z(x+d_{G}(v,v_i))
    = Z(2k-x)-2\sum_{i=1}^{Z}d_{G}(v,v_i)\\
& \geq Z(2k-x)-2Z(k-x)
    = Z(2k-x - 2k + 2x) = Zx
\end{align*}
Since $G$ forms a SUM-NE, this cannot be an IR, hence $Zx\leq p(x)$.
We get $p(x) \geq \frac{k-x}{\min\{p(x^*)+x^*, \hat\beta\}}x$, giving $k \leq \min\{p(x^*)+x^*, \hat\beta\}\cdot p(x)/x+x$.

If $\min\{p(x^*)+x^*, \hat\beta\} = \hat\beta$, then the diameter of $G$ is at most
$2(p(\hat\beta) + \hat\beta) \leq 2(p(\hat\beta) + p(x^*)+x^*) = \LDAUOmicron{p(x^*)+x^*}$.
Otherwise, if $\min\{p(x^*)+x^*, \hat\beta\} = p(x^*)+x^*$, then the diameter is at most $(p(x^*)+x^*)p(x)/x + x$.
For $p(x^*) \leq x^*$ the lemma follows by setting $x \coloneqq x^*$.
In case $p(x^*) > x^*$, the diameter is at most $\LDAUOmicron{(p(x^*)+x^*)\cdot p(p(x^*))/p(x^*)}$ by setting $x \coloneqq p(x^*)$.
Using the monotonicity of $p$, it holds $p(p(x^*)) \leq p(x^*)$ and we get $\LDAUOmicron{p(x^*)+x^*}$.
\end{proof}

\begin{theorem}
\label{thm:sumPoaUpperBound}
Let $p:[\check\beta,\hat\beta]\rightarrow\mathbb{R}^+$ be a price function and $x^*\in[\check\beta,\hat\beta]$ the value minimizing $p(x) + x$, then in the SUM-game
$\PoA=\LDAUOmicron{\min\{n, (p(x^*) + x^*)/\check\beta\}}$.
\end{theorem}
\begin{proof}
Let $G$ be an arbitrary SUM-NE graph, then by Lemma~\ref{lemma:SumNEboundByDiameter} the diameter is at most $\LDAUOmicron{p(x^*)  +x^*}$.
Hence, $n \delta_{G}(v) = \LDAUOmicron{n(n-1)(p(x^*)+x^*)}$ and we get by Lemma~\ref{le:neCost} that for every SUM-NE strategy profile the social cost is at most
$\LDAUOmicron{n(n-1)(p(x^*)+x^*)}$.
Further, with $\check\beta$ being the edge weight lower bound, Lemma~\ref{lemma:sumSocialCostBound} gives
$2\check\beta n(n-1) + m(p(x^*) + x^* - 4\check \beta)$ as social cost lower bound, whereas $m$ denotes the number of edges.

If $p(x^*)+x^* \leq 4\check\beta$, then the lower bound is minimized with $m=n(n-1)/2$ and becomes $n(n-1)(p(x^*) + x^*)/2$, which gives a PoA of $\LDAUOmicron{1}$.
Otherwise, if $p(x^*) + x^* > 4\check\beta$, the lower bound is minimized with $m=n-1$ and we get
$\PoA = \LDAUOmicron{\frac{n(n-1)(p(x^*)+x^*)}{(n-1)(2\check\beta n + p(x^*) + x^*-4\check\beta)}}$.
When separately considering whether $n < \frac{p(x^*) + x^*}{\check\beta}$ holds or not, we get $\PoA=\LDAUOmicron{\min\{n, (p(x^*) + x^*)/\check\beta\}}$.
\end{proof}

Applying the price and weight value ranges, we can deduce a price of anarchy upper bound that is independent of the price function.

\begin{corollary}
For every price function $p:[\check\beta,\hat\beta]\rightarrow\mathbb{R}^+$, in the SUM-game it holds
$\PoA = \LDAUOmicron{\min\{1 + p(\check\beta)/\check\beta, (p(\hat\beta) + \hat\beta)/\check\beta, n\}}$.
\end{corollary}

The price of anarchy upper bound is even tight for a broad class of price functions.
In particular, for all price functions $x\mapsto p(x)$ that decrease faster than the linear function $x\mapsto -x$ and where both $p(\hat\beta) \leq \check\beta$ and $p(\check\beta)\leq \hat\beta$ hold, it is
$\PoA=\Omega(\min\{1+p(\check\beta)/\check\beta, (p(\hat\beta) + \hat\beta)/\check\beta, n\})$.
Yet, the bound cannot be tight for every function as can be seen when considering $p:[1,1]\rightarrow [\alpha,\alpha]$, which constitutes the original game by \textcite{fabrikant2003}, for which it is known that for most ranges of $\alpha$ the price of anarchy is constant (cf.\ related work).

\begin{theorem}
\label{thm:sumNeUpperBound}
Let $p:[\check\beta,\hat\beta]\rightarrow\mathbb{R}^+$ be a price function with $p(\hat\beta) \leq \check \beta$, $p(\check\beta)\leq\hat\beta$, and $\hat\beta =\argmin_{x\in[\check\beta,\hat\beta]} p(x)+x$,
then $\PoA=\LDAUOmega{\min\{n, (p(x^*) + x^*)/\check\beta\}}$.
\end{theorem}
\begin{proof}
First, we show that a complete graph with every edge having weight of $\hat\beta$ forms a SUM-NE.
Let $S$ be a strategy profile such that $G[S]$ forms a complete graph where every edge has a weight of $\hat\beta$.
For an arbitrary node $u\in G[S]$ with strategy $s_u\in S$, we consider an improving response strategy $s'_u$.
Denote by $m \coloneqq |s_u|$ the number of edges of $u$ with strategy $s_u$ and by $m' \coloneqq |s'_u|$ the corresponding number of edges with the changed strategy.
Further, define $E_x \coloneqq \{\{u,v\} | (v,x) \in s'_u\}$ to be the edges of weight $x$ owned by $u$ in $s'_u$ and denote the lowest weight of any edge in $s'_u$ by $\check x$.

For every improving response it holds $m' \leq m$, since for any edge weight $x$ the costs for creating an edge exceed the possible gain, which is
$\hat\beta - (p(x)+x) < \hat\beta + p(\hat\beta) - (p(x) + x) \leq 0$.
Since $G[S]$ is a complete graph, with $s'_u$ the diameter is at most two hops and hence, the distance of $u$ to every not directly connected node is $\hat\beta + \check x$.
This gives as upper bound for the gain by changing to $s'_u$:
$$
    m(p(\hat\beta) + \hat\beta) - \left[\sum_{x\in[\check\beta,\hat\beta]}|E_x|(p(x) + x) + (\hat\beta + \check x)(m-m')\right]
    \leq (m-m')(p(\hat\beta) + \hat\beta -\hat\beta - \check x) \leq 0
$$
Hence, the stated graph forms a SUM-NE and by comparing its cost to the social cost of a star, with all edges having weight $\check\beta$, we get as PoA lower bound:
\begin{align*}
\PoA & \geq \frac{p(\hat\beta)n(n-1) + \hat\beta n (n-1)}{p(\check\beta)(n-1) + 2\check\beta(n-1)(n-2) + 2(n-1)\check\beta} \\
    & = \frac{n(p(\hat\beta) + \hat\beta)}{p(\check\beta) + 2\check\beta(n-1)}
        = \LDAUOmega{\min\{n, (p(\hat\beta) + \hat\beta)/\check\beta\}}
\end{align*}
\end{proof}

Concluding the analysis, we apply our results to explicit price functions.

\begin{corollary}
\label{cor:oneByWresults}
For the price function $p:[\check\beta,\hat\beta]\rightarrow\mathbb{R}^+, x \mapsto \alpha/x$ with $\alpha > 0$ the bounds $\PoS=\LDAUOmicron{1}$ and $\PoA=\LDAUOmicron{\sqrt{\alpha}/\check\beta}$ hold.
\end{corollary}

\begin{corollary}
\label{cor:linearFctResults}
For the price function $p:[1,\alpha - 2\varepsilon]\rightarrow\mathbb{R}^+, x \mapsto \alpha - (1 + \varepsilon) x$ with $\alpha > 0, \varepsilon\in (0, \nicefrac{1}{n-1})$ the bounds $\PoS=\LDAUOmicron{1}$ and $\PoA=\LDAUTheta{\alpha - \varepsilon}$ hold.
The optimal solution is given by a star with all edges having weight of $1$.
\end{corollary}

\section{The Max-Game}
In this section, we consider the quality of equilibria in the MAX-game.
\begin{lemma}
\label{lemma:sumSocialCostBoundMax}
Let $p:[\check\beta,\hat\beta]\rightarrow\mathbb{R}^+$ be a price function, $x^*\in[\check\beta,\hat\beta]$ the value minimizing $x + p(x)/2$, and $S$ a strategy profile.
Then, $c(S)\geq (x^* + p(x^*)/2)n$.
\end{lemma}
\begin{proof}
For every node $v_i\in V, i=1,\ldots,n$, we consider an arbitrary longest shortest path to any other node and denote the weight of the first edge of this path by $x_i$.
By this, there are at most two weights in the set $\{x_1,\ldots,x_n\}$ that correspond to the same edge and further, the price for each of those edges must be payed by one node.
Summing over all private costs, we get that $c(S) = \sum_{v\in V}c_v(S) \geq \sum_{v\in V} (x_v + p(x_v)/2) \geq (x^* + p(x^*)/2)n$.
\end{proof}

For a star with all edges having weight of $x^*$ and thus social cost $\LDAUTheta{x^*n+p(x^*)n}$, this bound is tight.

\begin{theorem}
\label{thm:maxGameNeGraphs}
For any price function $p:[\check\beta,\hat\beta]\rightarrow\mathbb{R}^+$, equilibrium graphs exist and the price of stability is constant.
\end{theorem}
\begin{proof}
Define $\chi^*\in [\check\beta,\hat\beta]$ to be the value minimizing $(n-1)p(x)+x$, and $\bar\chi\in[\check\beta,\hat\beta]$ the value minimizing $p(x)+x$.
Using these definitions, we get $\bar\chi \leq \chi^*$, $p(\bar\chi) \geq p(\chi^*)$, and $p(\bar\chi) + \bar\chi < (n-1)p(\chi^*) + \chi^*$.
\smallskip
\par
\noindent
\textit{(Stability of star with satellites owning all edges.)}
Consider $(n-1)p(\chi^*)+\chi^* \geq p(\bar\chi) + 2 \bar\chi$ and a star with all edges having weight of $\bar\chi$ and are owned by the satellites.
Then, the center node can only improve its private costs by creating $n-1$ edges of a weight $x < \bar\chi$, leading to a gain of $\bar\chi - ((n-1)p(x)+x) \leq \bar\chi + p(\bar\chi) - (p(x)+x) \leq 0$.
For any satellite $v$, there are two kinds of improving responses.
First, if $v$ only creates edges to satellites, the gain is at most
$p(\bar\chi) + 2 \bar\chi - ((n-2)p(x) + \bar\chi + x) \leq 0$.
Secondly, for a strategy change that connects the acting satellite to every other node, the gain is at most $p(\bar\chi)+2 \bar\chi - ((n-1)p(\chi^*)+\chi^*) \leq 0$.
Hence, the star graph forms a MAX-NE and by using the social cost lower bound we get

\begin{align*}
\PoS \leq \frac{(n-1) p(\bar\chi) + n 2\bar\chi}{(x^*+p(x^*)/2)n}
    \leq 2 \frac{(n-1)p(\bar\chi) + n2\bar\chi}{(\bar\chi + p(\bar\chi))n} \leq  4
\end{align*}

\smallskip
\par
\noindent
\textit{(Stability of star with center owning all edges.)}
On the other hand, if we have $(n-1)p(\chi^*)+\chi^* < p(\bar\chi) + 2 \bar\chi$, consider a star graph consisting only of edges of weight of $\chi^*$ owned by the center node.
If it further holds that $\chi^* \leq (n-2)p(\chi^*)$, we claim that this graph forms a MAX-NE.
By construction, $c$ cannot perform any improving response.
For any satellite $v$, we have to consider three kinds of best responses:
First, if $v$ creates edges of weight $x$ ($x\leq 2\chi^*$) to all other $n-2$ satellites, the gain is
$2\chi^* - (\max\{\chi^*, x\} + (n-2)p(x))$.
If $x < \chi^*$ this value is negative.
But also for $x \geq \chi^*$ we get a gain of at most $\chi^* + (n-2)p(\chi^*) - (x + (n-2)p(x)) \leq 0$, since the value minimizing $x + (n-2)p(x)$ lies in the interval $[\bar\chi,\chi^*]$ and thus the best possible choice for $x$ is $x=\chi^*$.
Secondly, if $v$ creates edges of weight $x$ to other nodes, the gain is
$2\chi^* - ((n-1)p(x) + x) \leq (n-2)p(\chi^*) + \chi^* - (n-1)p(\chi^*) - \chi^* \leq 0$.
Thirdly, if $v$ creates only one edge to $c$ of weight $x$, the gain is:
$2\chi^* - (p(x) + x + \chi^*)
    \leq \chi^* - p(\bar\chi) - \bar\chi
    \leq \chi^* - ((n-1)p(\chi^*) + \chi^* - \bar\chi)
    \leq (n-2)p(\chi^*) - (n-1)p(\chi^*) - \chi^* + \bar\chi < 0
$.
Hence, this star graph forms a MAX-NE and by using the social cost lower bound we get
\begin{align*}
\PoS \leq \frac{(n-1) p(\chi^*) + n 2 \chi^*}{(x^*+p(x^*)/2)n}
    \leq 2 \frac{(n-1)p(\chi^*) + n2 \chi^*}{(\bar\chi + p(\bar\chi))n}
    \leq 4 \frac{(n-1)p(\chi^*) + 2 \chi^*}{\bar\chi + p(\bar\chi)}
    \leq 8
\end{align*}

\smallskip
\par
\noindent
\textit{(Stability of clique with one node owning $n-1$ edges.)}
As the remaining case, we now consider $(n-1)p(\chi^*)+\chi^* < p(\bar\chi) + 2 \bar\chi$ and  $\chi^* > (n-2)p(\chi^*)$.
Here, we construct a star with one node $c$ owning $n-1$ edges of weight $\chi^*$ and complete this star to a clique with all edges having weight $\chi^*$, but arbitrary edge ownerships.
We claim that this graph forms a MAX-NE and at first note that by construction, $c$ has optimal weight for all of its edges.
Also every other node has optimal weight for its edges, since by unilaterally changing its edge weights the diameter stays at least $\chi^*$.
We further show that no node will change its edge set be considering the following kinds of best responses.
First, for any node exchanging its edges by edges to all other nodes, the optimal weight is $\chi^*$ and hence, doing so cannot improve its private costs.
Secondly, by simply removing all own edges the gain is at most $(n-2)p(\chi^*) - \chi^* < 0$.
Thirdly, by removing all own edges and creating one edge to $c$ of weight $x$, the gain is at most
$(n-1)p(\chi^*) + \chi^* - (\chi^* + x) - p(x)
    \leq (n-1)p(\chi^*) - x - p(x)
    \leq \chi^* + p(\chi^*) - (\bar\chi + p(\bar\chi))
    \leq 0
$.
Hence,
\begin{align*}
\PoS \leq \frac{p(\chi^*)(n-1)n/2 + n \chi^*}{(x^*+p(x^*)/2)n}
    \leq \frac{(n-1)p(\chi^*) + \chi^*}{(x^*+p(x^*))/2}
    \leq \frac{p(\bar\chi) + 2 \bar\chi}{(x^*+p(x^*))/2} \leq  4
\end{align*}
\end{proof}

\begin{lemma}
\label{lemma:MaxGameDiameterToPoA}
Let $p:[\check\beta,\hat\beta]\rightarrow\mathbb{R}^+$ be a price function with $x^*\in[\check\beta,\hat\beta]$ minimizing $x + p(x)/2$ and $S$ a MAX-NE strategy profile.
Then, $c(S) \leq n \delta_{G[S]}(v) + x^*(n-1) + 2(p(x^*)+x^*)(n-1)$,
with $\delta_{G[S]}(v)$ the distance costs of an arbitrary $v \in V$.
\end{lemma}
\begin{proof}
First, we show that in an equilibrium graph all edges have a price of at most $p(x^*)+x^*$.
Suppose to the contrary that there is an edge having price of $p(x) > p(x^*)+x^*$.
Then exchanging this edge with one of weight $x^*$ would decrease the edge costs by $p(x)-p(x^*) > x^*$ while increasing the distance costs by at most $x^*$.
Since this would contradict $S$ being a MAX-NE, no edge has weight of more than $p(x^*)+x^*$.

Next, take an arbitrary node $v$ and consider a shortest path tree $T$ of $v$ in $G$.
For every node $u\in V$ we define $m_u  \coloneqq  |\{\{u,v\} | (v,x) \in s_u \wedge \{u,v\} \in T\}|$ to be the number of tree edges that belong to $T$ and are maintained by node $u$.
Then, for the cost of any node $u \neq v$ it holds:
\[
    c_u(S) \leq (p(x^*)+x^*)(m_u + 1) + \delta_{G}(v)
\]
To see this, we define a strategy change for $u$ yielding exactly this value.
Since $S$ forms a MAX-NE, the costs of the changed strategy are an upper bound of the private costs of $u$.
We consider the strategy that is obtained when $u$ removes all own edges, except those belonging to $T$, and additionally creates one edge of weight $x^*$ to $v$.
This strategy incurs edge costs of at most $(p(x^*)+x^*)m_u + p(x^*)$ for node $u$.
By having an edge of weight $x^*$ to $v$, the distance costs are bound by at most $\delta_{G}(v)+x^*$, since the strategy change does not change the distances of $v$ to any other node.

Using the private costs bound for every node $u \ne v$ and the fact that $v$ only owns edges belonging to $T$ (otherwise removing a non-tree edge would improve $v$'s cost), we have:
\begin{align*}
c(S) &\leq \delta_{G}(v) + (p(x^*)+x^*) m_v + \sum_{u \neq v}[(p(x^*)+x^*)(m_u + 1) + \delta_{G}(v) + x^*]\\
    &= n \delta_{G}(v) + x^*(n-1) + (p(x^*)+x^*) m_v +\sum_{u \neq v}{(p(x^*)+x^*)(m_u + 1)} \\
    & = n\delta_{G}(v) + x^*(n-1) + 2 (p(x^*)+x^*)(n-1).
\end{align*}
For the last equality we use that the number of edges in a tree of $n$ nodes is $n-1$.
\end{proof}

Using a similar approach like \textcite{demaine2007}, we derive a bound for the diameter and hence for the social cost of every MAX-NE graph.

\begin{lemma}
\label{lemma:MaxNeDiameterBound}
Let $p:[\check\beta,\hat\beta]\rightarrow\mathbb{R}^+$ be a price function and $S$ a strategy profile forming a MAX-NE.
Then, the diameter of $G[S]$ is at most $\mathcal{O}(\sqrt[3]{p(x)^2xn} + x)$, for $x \in [\check\beta,\hat\beta]$ arbitrary.
\end{lemma}
\begin{proof}
Given an arbitrary node $u\in V$ and an edge weight $x \in [\check\beta,\hat\beta]$, we define $N_{k}(u)$ to be the nodes within distance of at most $kx \leq \delta_{G[S]}(u) / 2$ to $u$.
We claim that then for $k\leq \delta_{G[S]}(u)/(2x)$ it holds $|N_k(u)|\geq \frac{x}{2 p(x)} (k^2 - 3k + 2)$.
For this, consider a shortest-path tree $T$ rooted at $u$ and define $t$ as its number of leaves.
Let $P_1, P_2, \ldots, P_t$ denote the shortest paths from $u$ to all leaves.
Moreover, for every $k$ define $Q_k$ to be the set containing the first node of every path $P_i$ that has distance of more than $kx$ and at most $(k+1)x$ to $u$.
Now consider a strategy change where $u$ creates an edge of weight $x$ to each node in $Q_k$.
If $k\leq \frac{\delta_{G[S]}(u)}{2x}$, then this decreases the distance to all nodes in $Q_k$ by at least $(k-1)x$.
The increase of $u$'s edge costs is $p(x)|Q_k|$.
Since $S$ forms a NE, it must hold $p(x)|Q_k| \geq (k-1)x$ and thus $|Q_k| \geq (k-1)x / p(x)$.
By $|N_{k}(u)| \geq \sum_{i=1}^{k-1} |Q_i| \geq \sum_{i=1}^{k-1} (i-1)x / p(x) = \frac{x}{2 p(x)} (k^2-3k+2)$, this gives the lower bound.

Next, consider a node $v\in V$ with maximal distance costs, i.e., $\delta_{G[S]}(v) = \diam(G[S])$.
We select a set of cluster centers $C$ such that two conditions hold: first, for any node $u \in V$ the minimal distance to a center is $d_{G[S]}(u,C) \leq 2kx$; secondly, the distance between any two centers $c,c' \in C, c \neq c'$ is $d_{G[S]}(c,c') > 2kx$.
We set $k  \coloneqq  \frac{\diam(G[S])-x}{4x}$ and observe that $k \leq \frac{\delta_{G[S]}(c)}{2x}$ for all centers $c$, since otherwise there is a center $c'$ having distance of not more than $2kx$ to all nodes, implying a smaller diameter than $\diam(G[S])$.
By construction of $C$, $n \geq \sum_{c \in C}|N_{k}(c)| \geq \frac{|C| x}{2 p(x)}(k^2-3k+2)$.
Considering a strategy change where $v$ buys edges of weight $x$ to all $c \in C$, this would decrease $v$'s distance costs by at least $\delta_{G[S]}(v) - (2k + 1)x \geq 2kx$ while increasing the edge costs by $|C|p(x)$.
Since $S$ forms a MAX-NE, we have $|C| p(x) \geq 2kx$ and obtain:
\begin{align*}
n \geq |C| \cdot \frac{x}{p(x)} \cdot \frac{k^2-3k+2}{2} \geq \frac{2kx}{p(x)} \cdot \frac{x}{p(x)} \cdot \frac{k^2-3k+2}{2}
\end{align*}
This gives $0 \geq k^3 -3 k^2 + 2k - \frac{p(x)^2n}{x^2}$ and hence $k = \LDAUOmicron{\left(p(x)^2n / x^2\right)^\frac{1}{3}}$.
\end{proof}

\begin{theorem}
Let $p:[\check\beta,\hat\beta]\rightarrow\mathbb{R}^+$ be a price function.
Then, the price of anarchy in the MAX-game is bounded by $\LDAUOmicron{1 + \sqrt[3]{n}}$.
\end{theorem}
\begin{proof}
Let $x^*\in [\check\beta,\hat\beta]$ be the value minimizing $p(x) + x/2$.
Then for every MAX-NE strategy profile $S$, we get with Lemma~\ref{lemma:MaxGameDiameterToPoA} and Lemma~\ref{lemma:MaxNeDiameterBound} (using $x \coloneqq x^*$) the social cost upper bound
$c(S)\leq n(p(x^*)+2x^*+\sqrt[3]{p(x^*)^2x^*n})$.
Comparing this to the social cost lower bound from Lemma~\ref{lemma:sumSocialCostBoundMax}, we get
$\PoA = \mathcal{O}((p(x^*)+2x^*+\sqrt[3]{p(x^*)^2x^*n})/(x^* + p(x^*)))$, which can be simplified to
$\PoA = \mathcal{O}\left(1 + \sqrt[3]{n}\cdot\min\left\{\sqrt[3]{x^*/p(x^*)}, \sqrt[3]{(p(x^*)/x^*)^2}\right\}\right)$ and gives the claim by case distinction on whether $x^*/p(x^*) < 1$ holds or not.
\end{proof}

\section{Conclusion}
Our model extension captures the effects of quality of service agreements in autonomous distributed networks and provides theoretical results regarding the stable states of these networks.
Interestingly, despite of the considerably increased freedom in the strategic decisions of the nodes (e.g., for a continuous price function the strategy set is unbounded), equilibria always exist.
In the SUM-game for a given price function $p$, we discovered the value that minimizes the term $p(x) + x$ to characterize the worst case loss by selfish behavior.
This value can be understood as the optimal trade-off for using one edge for exactly one shortest path and effectively bounds the maximal investment into any edge.

\begin{sloppy}
\printbibliography
\end{sloppy}

\end{document}